\documentclass[12pt,a4paper]{article}

\RequirePackage{amsthm,amsmath,amsfonts,amssymb}
\RequirePackage[colorlinks,citecolor=blue,urlcolor=blue]{hyperref}
\RequirePackage{graphicx}
\RequirePackage{float}

\newtheorem{theorem}{Theorem}

\newtheorem*{remark}{Remark}

\def\P{\mathbb{P}}
\def\F{\mathbb{F}}
\def\E{\mathbb{E}}
\def\N{\mathbb{N}}

\usepackage{tikz}

\usepackage{makecell}
\usepackage{hyperref} 

\usetikzlibrary{automata, positioning, arrows}

\tikzset{
->, 
>=stealth', 
node distance=3cm, 
every state/.style={thick, fill=gray!10}, 
initial text=$ $, 
}

\title{
Estimating Transition Rates in Two-State Non-Homogeneous Markov Jump Processes with Intermittent Observations:
\\

A Pseudo-Marginal McMC Approach via Honest Times 
}

\author{Dario Gasbarra
\\ Department of Mathematics and Statistics, \\ University of Vaasa, Finland, 
\\  \\ Sangita Kulathinal,  Etienne Sebag
\footnote{Corresponding author: etienne.sebag@helsinki.fi}
\\ Department of Mathematics and Statistics, \\ University of Helsinki,  Finland}

\begin{document}

\maketitle

\begin{abstract}
A possibly time-dependent transition intensity matrix or  generator $(Q(t))$  characterizes the law of a Markov jump process (MP). For a time homogeneous MP, the transition probability matrix (TPM) can be expressed as a matrix exponential of $Q$. However, when dealing with a time non-homogeneous MP, there is often no simple analytical form of the TPM in terms of $Q(t)$, unless they all commute. This poses a challenge because when a continuous MP is observed intermittently, a TPM is required to build a likelihood. 
In this paper, we show that the estimation of the transition intensities of a two-state non-homogeneous Markov model can be carried out by augmenting the intermittent observations with honest random times associated with two independent driving Poisson point processes, and that
sampling 
the full path is not required. We propose a pseudo-marginal McMC algorithm to estimate the transition rates using the augmented data. Finally,  we illustrate our approach by simulating a continuous MP and by using observed (intermittent) time grids extracted from real clinical visits data. 
\end{abstract}


\section{Introduction}
  Multistate models are commonly employed in many medical applications to study diseases within a cohort, whereby different states describe specific statuses, or markers, of a disease. In this framework, researchers are interested in the state occupancy times (how long individuals remain in each state) and the exact transition times (when individuals move between states). This information is essential to better understand and track disease progression or remission, as well as to identify how individuals respond to a certain treatment, for example. Nonetheless, continuous-time state observations are rarely available in practice because the monitored subjects having the disease are only observed at fixed points in time corresponding to clinical visits. The set-up naturally lends itself to that of a non-homogeneous Markov process, characterized by movements between states that depend on time. The process, $\{X(t), t \ge 0\}$, is observed at intermittent, discrete time-points for each individual. Therefore, clinicians/practitioners only have access to snapshots in time concerning the underlying process instead of having the full trajectory.  \\
 
The underlying dynamics of the process are often captured in terms of the transition intensities, and the transition probabilities are then derived from the intensities by solving Kolmogorov forward equations. From a statistical point of view, when considering such a process, the interest is in the estimation of these parameters, as well as in the prediction of the latent path sequence at unobserved times. When we know the exact state transition times, then the likelihood is written in terms of the transition intensities, and is tractable and easy to work with. Yet in the more common cases that the occupied states are only known at specific and intermittent time points, the likelihood must be constructed by considering the transition probabilities over intervals. Since the states are only known at specific times, e.g. clinical visits, say $t_{1}$ and $t_{2}$, and what happened between $t_{1}$ and $t_{2}$ is unknown, one must rely exclusively on the transition probabilities $P(X(t_2) = j \mid X(t_1) = i)$
to formulate the likelihood of transitioning between states over the unobserved time interval. 
The overarching problem of using these transition probabilities in complex multistate models is that, often, the Kolmogorov forward equations do not have nice closed-form analytical solutions, which can consequently impose a substantial numerical and computational burden under practical implementations. Various workarounds exist in the statistical literature; approximating the transition intensities by piecewise constant intensities and then applying the matrix-exponential, using a smooth function of time, or using a time transform on which the process becomes homogeneous \cite{Titman2011, Hubbard2007, Kendall2024}. Some of the methods are  made available in the \texttt{nhm}  \texttt{R}-package (\cite{nhmTitman2019}). In general, these approaches have limited applicability. \\


In this paper, we propose a pseudo-marginal Markov chain Monte Carlo (McMC) approach which consists of a novel data augmentation procedure using Poisson point processes (PPP) to eliminate the need altogether to handle these computations. In order to illustrate this methodology, we will focus on a simple two-state reversible model to capture the dynamics of disease in which subjects transition back-and-forth between two distinct health states. This is illustrated in Figure~\ref{2_state} and examples can be found from \cite{Mehtala2011, Cook2018}. It is of statistical interest to estimate the transition intensities $\lambda_{0}(t)$ and $\lambda_{1}(t)$ for this model. Two-state disease models often deal with binary clinical hard outcomes which are directly measurable at each clinical visit. For example, a hard clinical outcome at each visit might be the presence (state 1) or absence (state 0) of a disease marker which subsequently can be used to allow for broader categorizations such as ``infected'' vs ``not-infected'', or ``disease onset'' vs ``no disease''. This disease marker may fluctuate at any point in time. Moreover, we assume that each subject's disease status is fully captured by one of these two states at any given point in time. Treatment for diseases may, or may not, be administered at each clinical visit that aim to influence these transitions. For instance, if a subject stays in the disease state 1 for an extended period of time, then we can speculate that this subject is exhibiting a poor response to the treatment.
\\

Nowadays, there exists several statistical packages in \texttt{R} to assist with statistical inference of multistate models, namely the \texttt{msm} package \cite{Jackson2011}, \texttt{nhm}, and \texttt{mstate} packages. The \texttt{nhm} package, which builds on \texttt{msm}, is tailored to estimating transition rates from panel data. Nevertheless, to our knowledge, the computations fail under \texttt{nhm} when the observed time grids are irregular; that is, when there is variation in the time grids across subjects. Hence, our proposed methodology aims to address this gap by providing a versatile framework that handles estimation in non-homogeneous cases for two-state models. The rest of the paper is organized as follows: in Sections~\ref{sec:theory} and~\ref{sec:unbiased_tpm} we present the theoretical background in order to explain how the intractable computations in this setting arise and we motivate the use of PPPs and honest random times. We highlight why the data augmentation procedure involving a PPP is a contextually and mathematically justifiable framework. In Section~\ref{sec:data_aug}, we develop the specific details of the pseudo-marginal McMC algorithm, which involves a data augmentation step. Next, in Section~\ref{sec:example_sim} , we carefully illustrate the steps and computations related to this algorithm by using as an example Weibull (time-dependent) rates. Section~\ref{sec:results} depicts the results from our algorithm when applied to a simulated process and when using real-life subject time grids. Both scenarios employ Weibull-distributed rates. We end with a discussion in Section~\ref{sec:discussion}. 


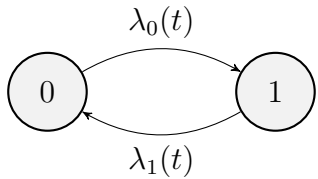
\begin{figure}[h]
\centering
\begin{tikzpicture}
\node[state] (1) {$0$};
\node[state, right of=1] (2) {$1$};
\draw (1) edge[bend left, above] node{$\lambda_{0}(t)$} (2)
(2) edge[bend left, below] node{$\lambda_{1}(t)$} (1);
\end{tikzpicture}
\caption{A two-state model describing a typical process of a disease marker}
\label{2_state}
\end{figure}

\section{Theoretical Background}\label{sec:theory}
When we only have access to the intermittent observations of the occupied state at clinical visits, the likelihood formulation is in terms of the transition probabilities. We denote by $\{X(t), t \ge 0\}$ the observed process.  Furthermore, let $\mathcal{H}(t) = \{X(s), 0 \leq s \leq t\}$ denote the corresponding history (collection of available information) of the process over the time interval $[0,t]$. 

\subsection{Transition intensities and probabilities}
Both the transition intensities and the transition probabilities govern the stochastic movement between states and they exhibit an explicit dependence on the history $\mathcal{H}(t).$ The latter take the general form, for $0 \leq s <t$, 
$$P_{ij}(s,t) = P_{ij}(s,t \ |\mathcal{H}(s)) = \newline P\big(X(t) = j \mid X(s) = i, \mathcal{H}(s)\big).$$ 
Recall also that the transition intensities $\lambda_{0}(t \ | \mathcal{H}(t-))$ and $\lambda_{1}(t \ | \mathcal{H}(t-))$ need to be estimated for the model. Moving forward, for the sake of convenience, we will suppress the dependency on $\mathcal{H}(t)$ in the notation.

The generator matrix $\boldsymbol{Q}(t)$, also called the transition intensity matrix, for a two-state  model (Figure \ref{2_state}) is given by
\begin{equation}
  \boldsymbol{Q}(t) =
\begin{bmatrix}
    -\lambda_{0}(t) & \lambda_{0}(t) \\
        \lambda_{1}(t) & -\lambda_{1}(t)
\end{bmatrix}.
 \label{Q_matrix}
\end{equation}

\noindent The corresponding transition probability matrix is denoted by $\boldsymbol{P}(s,t)$ and has the form
\begin{equation}
  \boldsymbol{P}(s,t) =
\begin{bmatrix}
    P_{00}(s,t) & \hspace{5pt}P_{01}(s,t)\\
        P_{10}(s,t) & \hspace{5pt} P_{11}(s,t)
\end{bmatrix}.
 \label{P_matrix}
\end{equation}

\noindent Note that the rows of $\boldsymbol{Q}(t)$ sum to zero while the rows of $\boldsymbol{P}(s,t)$ must sum to one. The transition probabilities satisfy the linear ordinary differential equations (ODE) obtained using the Kolmogorov forward equations as

\begin{equation}
\frac{\partial}{\partial t}
                    \label{ode}\boldsymbol{P}(s,t) = \boldsymbol{P}(s,t)\boldsymbol{Q}(t).
\end{equation}

\noindent Performing the matrix multiplication expansion on the right hand side, the $(1,2)$ and $(2,1)$ matrix entries, for example, respectively yield

\begin{align*}
    \frac{\partial}{\partial t}P_{01}(s,t) &= P_{00}(s,t)\lambda_{0}(t)  \  - P_{01}(s,t)\lambda_{1}(t), \\
    \frac{\partial}{\partial t}P_{10}(s,t) &= - P_{10}(s,t)\lambda_{0}(t) \  + P_{11}(s,t)\lambda_{1}(t).
    \end{align*}
    
\noindent Solutions to the above ODE's, and use of the fact that the transition probabilities out of a state sum to $1$, allows expressing the transition probabilities fully in terms of the transition intensities, with result:

\begin{align}
    P_{01}(s,t) & =\int_s^t \exp\biggl( - \int_v^t \bigl( \lambda_{0}(r) + \lambda_{1}(r) \bigr) dr \biggr) \lambda_{0}(v) dv, \nonumber \\
    P_{00}(s,t) & = 1 - P_{01}(s,t), \nonumber \\
    P_{10}(s,t) & 
     = \int_s^t \exp\biggl( - \int_v^t \bigl( \lambda_{0}(r) + \lambda_{1}(r) \bigr) dr \biggr) \lambda_{1}(v) dv \nonumber \\ 
   P_{11}(s,t) & = 1 - P_{10}(s,t) = 1-\int_s^t \exp\biggl( - \int_v^t \bigl( \lambda_{0}(r) + \lambda_{1}(r) \bigr) dr \biggr) \lambda_{1}(v) dv \nonumber \\
   & = \exp\biggl( - \int_s^t \bigl(\lambda_{0}(r) + \lambda_{1}(r) \bigr) dr \biggr) 
 +  \int_s^t \exp\biggl( - \int_v^t \bigl( \lambda_{0}(r) + \lambda_{1}(r) \bigr) dr \biggr) \lambda_{0}(v) dv. \nonumber \\ \label{TPM}  
 \end{align}

The expression for $P_{01}(s,t)$ can be simplified to some extent as shown below.
\begin{align}
P_{01}(s,t) & =  \int_s^t \exp\biggl( - \int_v^t \bigl( \lambda_{0}(r) + \lambda_{1}(r) \bigr) dr \biggr) \lambda_{0}(v) dv \nonumber  \\
& = \int_s^t   \exp\left(-\int_v^t \lambda_{1}(r) \,dr\right)  \; d\exp\left(-\int_v^t \lambda_{0}(r) \,dr\right).
\label{prb_expression}
\end{align}

\noindent  The second equality follows because 
$$
\exp\biggl(-\int_v^t \lambda_{0}(r)\, dr\biggr) \lambda_{0}(v)\, dv 
= d\exp\biggl(-\int_v^t \lambda_{0}(r)\, dr\biggr).
$$ 

\noindent  As we will see in Section~\ref{sec:example_sim}, the equations for the transition probabilities simplify nicely under the case of constant transition intensities $\lambda_{0}(t) = \lambda_{0}$ and $\lambda_{1}(t) = \lambda_{1}$.  However, the overarching problem is that, when we deal with time-dependent (non-homogeneous) transition rates, such as those that are of Weibull-type, then we still cannot compute analytically the integral in~(\ref{prb_expression}). Although there may be a possibility of finding some special function to represent the integral, this paper introduces and describes an alternative framework using Poisson point processes (PPP) for data augmentation which entirely avoids the need to compute the transition probabilities analytically. 

\subsection{PPP, honest times and data augmentation}\label{sec:honest}
The foundational approach of this distinct framework is to obtain an unbiased estimator of the transition probabilities, and thereby of the observed data likelihood, under this two-state setting. This is achieved by performing a data augmentation procedure by using a simpler PPP, instead of attempting to numerically solve the integrals for the transition probabilities. Consider an independent PPP $\Pi^1(dt)$ with intensity $\lambda_{1}(t)$. Then, by noting that the integrand in Equation~(\ref{prb_expression}) is the survival probability, integral~(\ref{prb_expression}) can be written as:

\begin{align}
(\ref{prb_expression}) & =
  \int_s^t P\left( \Pi^1\big((v,t]\big)=0\right)
 d\exp\left(-\int_v^t \lambda_{0}(r) \,dr\right) \nonumber \\
& = \mathbb{E}_{\Pi^1}\biggl( 
 \int_s^t {\bf 1}( \Pi^1( (v,t])=0) d\exp(-\int_v^t \lambda_{0}(r) dr)
\biggr) \; 
  \nonumber \\ &
= \mathbb{E}_{\Pi^1}\biggl( 
 \int_{\tau^1}^t  d\exp(-\int_v^t \lambda_{0}(r) dr)
\biggr) 
= 1- \mathbb{E}_{\Pi^1}\biggl( 
\exp(-\int_{\tau^1}^t \lambda_{0}(r) dr)
\biggr), \nonumber \\
 & \label{unbiasedTP}
 \end{align}
where  
$\tau^1 = s \vee \inf \{ r:   \Pi^1((r,t])=0 
\}$ is a {\it honest} random time depending only on the PPP $\Pi^1$ (see Appendix \ref{appC}). 
Fubini's theorem has been applied between the first and the second equality operator, interchanging the expectation and the integral sign. \\

Similarly, we construct an independent PPP with intensity $\lambda_{0}(t)$ and define $\tau^0 = s \vee \inf \{ r:   \Pi^0((r,t])=0\}$, a honest random time depending only on the PPP $\Pi^0$ so that 
\begin{align}
& \mathbb{E}_{\Pi^0}\biggl( 
 \int_s^t {\bf 1}( \Pi^0( (v,t])=0) d\exp(-\int_v^t \lambda_{1}(r) dr)
\biggr) \;   \nonumber \\ 
&
= \mathbb{E}_{\Pi^0}\biggl( 
 \int_{\tau^0}^t  d\exp(-\int_v^t \lambda_{1}(r) dr)
\biggr) 
 = 1- \mathbb{E}_{\Pi^0}\biggl( 
\exp(-\int_{\tau^0}^t \lambda_{1}(r) dr)
\biggr). 
  \label{unbiasedTP:rho}
 \end{align}

The honest random times $\tau^1$ and $\tau^0$ are independent with respective probability distributions
\begin{align} 
f^{(k)}( \tau^k; \theta)  & =  \lambda_{k}(\tau^k)^{\eta^k} 
\exp\biggl( - \int_{\tau^k}^{t}
 \lambda_{k}(s) ds \biggr), {\bf 1 }(s \le \tau^k < t) \label{eq:fi_1}, 
\end{align} 
where $\theta$ is the collection of the transition rates and $\eta^k = {{\bf 1}(\tau^k > s)}, k = 0, 1$.

\section{Unbiased estimators of the transition probabilities}\label{sec:unbiased_tpm}
\subsection{Using the PPP $\Pi^1$ with intensity $\lambda_{1}(t)$}
Given a vector of honest random times $\tau^1$ depending only on the PPP $\Pi^1$ with intensity $\lambda_{1}(t)$ and the transition rates,  we introduce conditional transition probabilities as:
\begin{align}
P_{00}^{\tau^1}(s,t)
&=\exp\biggl( - \int_{\tau^1}^{t}
\lambda_{0}(r) dr
\biggr), 
\; P_{01}^{\tau^1}(s,t) =  1- P_{00}^{\tau^1}(s,t), 
\nonumber\\ 
P_{11}^{\tau^1}(s,t)
&= \exp\biggl( - \int_{s}
^{t} ( \lambda_{0}(r) +
\lambda_{1}(r) ) dr \biggr)  + P_{01}^{\tau^1}(s,t), \nonumber\\
P_{10}^{\tau^1}(s,t) & = 1-P_{11}^{\tau^1}(s,t).
 \label{cond_sigma_eqs}
\end{align}
\begin{theorem}\label{Th1}
The conditional transition probability matrix  $P^{\tau^1}(s,t)$
is an 
unbiased estimator of $\mathbf{P}(s,t)$ defined in (\ref{TPM}). Further, the variance of each estimator is given by
\begin{align*} &
 \text{Var}\biggl( P_{k\ell}^{\tau^1}
 (s,t) \biggr)
 = \P( \tau^0_1 \vee \tau^0_2 \le \tau^1)
 - \P( \tau^0 \le \tau^1)^2,
&\end{align*}
where $\tau^1$ and $\tau^0$ are the honest times associated on the
interval $[s,t)$
with two independent PPP with rates $\lambda_{1}(t)$ and $\lambda_{0}(t)$, respectively.
\end{theorem}
\begin{proof}
The unbiasedness of the estimators given in (\ref{cond_sigma_eqs}) $P_{k\ell}^{\tau^1}( s,t)$ $k ,\ell \in \{0,1\}$  follow from (\ref{unbiasedTP}) so that 
\begin{align*}
\E_{\Pi^1}\bigl[ 
P^{\tau^1}_{k\ell}
( s,t) \bigr]
= P_{k\ell}(s,t),
\quad  k ,\ell \in \{0,1\}. 
\end{align*}

It is clear  that $\forall \; k,\ell\in \{0,1\}$, the conditional transition probabilities $P_{k\ell}^{\tau^1}( s,t)$ as  functions of the honest time $\tau^1$ have the same
variance. By conditioning and using the independence of $\tau^1$ and $\tau^0$ we obtain the second moment of the estimator $P_{00}^{\tau^1}(s,t)$ as follows.
\begin{align*}
&\P \bigl(  \tau^0_1 \vee \tau^0_2 
\le \tau^1 \bigr)
= \E\bigl[   
\P\bigl( \tau^0_1\vee \tau^0_2 
\le \tau^1 \big\vert \tau^1
\bigr) \bigr]
& \\ &
= \E\bigl[ 
\P\bigl(    \tau^0 \le 
\tau^1 \big\vert \tau^1\bigr)
\P\bigl( \tau^0_2 
\le \tau^1 \big\vert \tau^1
\bigr) \bigr]  & \\ & =
\E\biggl[\exp\biggl( - \int_{\tau^1}^{t_1} \lambda_{0}(r) dr\biggr)^2 \biggr]
= 
\E\bigl[ P_{00}^{\tau^1}(s,t)^2 \bigr] &
\end{align*}
while
$\E\bigl[ P_{00}^{\tau^1}( s,t) \bigr]
= \P( \tau^0 \le \tau^1).$
Hence the theorem.
\end{proof}

\subsection{Using the PPP $\Pi^0$ with intensity $\lambda_{0}(t)$}
By interchanging the role of $0$ and $1$ in 
(\ref{cond_sigma_eqs})
with honest random times $\tau^0$
depending only on the PPP $\Pi_0$
with intensity $\lambda_0(t)$,
we obtain alternative
unbiased estimators
of the transition probabilities:
\begin{align*}
P_{11}^{\tau^0}(s,t) & = \exp\biggl( - \int_{\tau^0}^{t} \lambda_{1}(r) dr \biggr), \; 
P_{10}^{\tau^0}(s,t)  =  1-  P_{11}^{\tau^0}(s,t), \nonumber\\ 
P_{00}^{\tau^0}(s,t) & = \exp\biggl( - \int_{s} ^{t} ( \lambda_{0}(r) +
\lambda_{1}(r) ) dr \biggr)  + P_{10}^{\tau^0}(s,t), \nonumber  \\
P_{01}^{\tau^0}(s,t) & = 1-P_{00}^{\tau^0}(s,t). 
\end{align*}

A theorem similar to Theorem \ref{Th1} can be given in terms of $\tau^0$ and its proof follows the same line of thinking. 

\begin{theorem}\label{ThmRho}
The conditional transition probability matrix  $P^{\tau^0}(s,t)$ is an 
unbiased estimator of $\mathbf{P}(s,t)$ defined in (\ref{TPM}). Further, the variance of each estimator is given by
\begin{align*}
 &\text{Var}\biggl( P_{k\ell}^{\tau^0}(s,t) \biggr)
 = \P( \tau^1_1 \vee \tau^1_2
 \le \tau^0)
 - \P( \tau^1 \le \tau^0)^2  & 
\end{align*}
where $\tau^1$ and $\tau^0$ are the honest times associated in the interval $[s,t)$ with two independent PPPs with rates $\lambda_{1}(t)$ and $\lambda_{0}(t)$, respectively.
\end{theorem}

\begin{remark}
Given both vectors $\tau^1$ and $\tau^0$, 
we can provide another unbiased estimator of the conditional transition probability matrix
\begin{align*}
P^{\tau^1,\tau^0}_{00}(s,t)
&=\exp\biggl( - \int_{\tau^1}^{t}
\lambda_{0}(r) dr
\biggr), 
P^{\tau^1,\tau^0}_{01}(s,t)
=  1 - 
P^{\tau^1,\tau^0}_{00}(s,t), \nonumber\\
P^{\tau^1,\tau^0}_{11}(s,t)
&=\exp\biggl( - \int_{\tau^0}^{t}
\lambda_{1}(r) dr
\biggr), 
P^{\tau^1,\tau^0}_{10}(s,t)
=  1- 
P^{\tau^1,\tau^0}_{11}(s,t).
\end{align*}
\end{remark}

\begin{remark}
    To reduce the variance we could also take average of the unbiased estimators of the transition probabilities with several independent copies of $\tau^1$ and or $\tau^0$, possibly choosing 
the representation with smaller variance.
\end{remark}

\subsection{Using $(\tau^1, \tau^0)$ pairs}\label{TPM:pair}
We have two independent PPPs, $\Pi^1$ with intensity $\lambda_{1}(t)$ and $\Pi^0$ with intensity $\lambda_{0}(t)$.  $\tau^1$ and $\tau^0$ are as defined in section \ref{sec:honest} and as illustrated in Figure \ref{fig:honest}. We can then write the transition probabilities (see equation (\ref{prb_expression})) using the pair $(\tau^{1}, \tau^{0})$ as:
\begin{align*} 
 P_{01}(s,t) &=  
  \int_s^t   \exp\biggl(-\int_v^t \lambda_{1}(r) dr\biggl)  \; d\exp\biggl(-\int_v^t \lambda_{0}(r) dr\biggl) & \nonumber \\ 
& =  1- 
\P_{\Pi^0\otimes \Pi^1}
\biggl(\tau^0 \le \tau^1 \biggr)
 =\P_{\Pi^0 \otimes \Pi^1 }\biggl(\tau^0 > \tau^1 \biggr) & \nonumber \\ 
 P_{00}(s,t) & = 1 - P_{01}(s,t) = 
 \P_{\Pi^0\otimes \Pi^1}
 \biggl(\tau^1 \ge \tau^0 \biggr) & \nonumber \\
 P_{10}(s,t) & =  
 \P_{\Pi^0 \otimes \Pi^1}
 \biggl(\tau^1 > \tau^0 \biggr)  & \nonumber \\
  P_{11}(s,t) & = 1 - P_{10}(s,t) =
  \P_{\Pi^0 \otimes \Pi^1}
  \biggl(\tau^0 \ge \tau^1 \biggr). & \nonumber  
\end{align*}
It is obvious that 
$P_{01}(s,t) <  P_{11}(s,t) \; \text{and} \; P_{10}(s,t) <  P_{00}(s,t).$ Note that 
\begin{align*}
 \bigl\{ \tau^0 = \tau^1 \bigr \}
 = \bigl\{ \tau^0= \tau^1=s\bigr\}
 = \bigl \{  \Pi^0( (s,t])=
 \Pi^1( (s,t]) = 0  \bigr \}  
\end{align*}
up to a set of probability zero
under the product distribution.

We can now give unbiased estimators of the transition probabilities as indicator functions of the inequalities.
\begin{align} 
& \widetilde P_{01}(s,t)  =  {\bf 1}(\tau^0 > \tau^1),  \;\widetilde P_{00}(s,t)  =   {\bf 1} (\tau^1 \ge \tau^0 ) \nonumber \\ 
& \widetilde P_{10}(s,t)  =   {\bf 1} (\tau^1 > \tau^0 ), \; \widetilde P_{11}(s,t)  =   {\bf 1} (\tau^0 \ge \tau^1 ). 
\label{TPM:both}
\end{align}

\begin{figure}
    \centering
    \includegraphics[width=12cm, height=8cm]{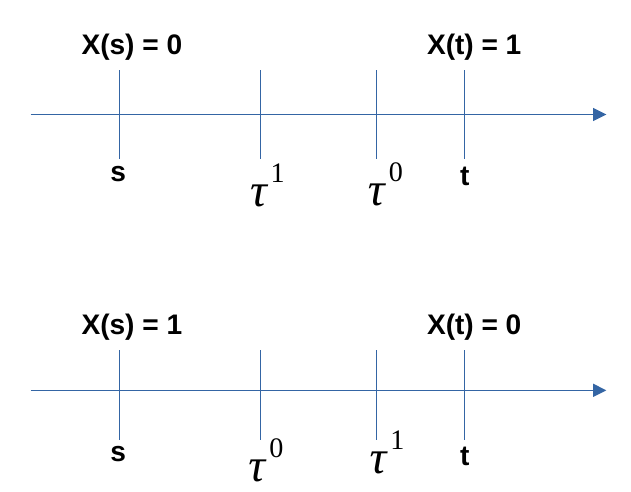}
    \caption{Honest times of the two independent Poisson point process over the interval $[s,t)$ for a two-state Markov process.}
    \label{fig:honest}
\end{figure}

We refer to Appendix \ref{subsec:appendix} for further insight into the properties of the two PPPs. We use (\ref{TPM:both}) in the data augmentation algorithm for estimation of the unknown rates.

\section{Data augmentation algorithm for an intermittently observed process}\label{sec:data_aug}
We now use the theory developed in section \ref{TPM:pair} for a fixed time grid $0=t_0< t_1< \dots < t_n$. Let $\Pi^k$ be a PPP  with intensity $\lambda_{k}(t)$ on the interval $(t_0, t_n]$ and for each consecutive times let 
\begin{align*}
\tau^k_i = t_{i-1} \vee \inf\{r:
 \Pi^k((r, t_i])=0 \},
 \quad i=1,\dots,n, k = 0, 1.
\end{align*}

\noindent Note that these $\tau^k_i$'s  are independent across $i$ and
we do not need to sample all the PPP $\Pi^k$; it is sufficient to sample the independent random times $\tau^k_i$ for each $i$. From the above definition, for example, of $\tau^1$, remark that $P\big(\Pi^1(r,t_i]=0\big) = \exp(-\int_{r}^{t_{i}} \lambda_{1}(s)ds)$, and this precisely has uniform distribution in $(0,1]$. The cumulative hazard function $\Lambda_1(r, t_i] = \int_{r}^{t_{i}} \lambda_{1}(s)ds$ represents the expected number of transitions from $1$ to $0$ in $(r, t_i]$. 
Also note that $P\big(\Pi^1(r,t_i]=0\big) = 0$ implies that $P\big(\Pi^1(s,t_i]=0\big) = 0, \; \forall s \ge r.$ 

As a direct consequence, we then have the following sampling scheme for $\tau^1_{i}$'s and $\tau^0_{i}$'s: for $y^k_i$ independent 1-exponential random variables, define 
\begin{align} &
\tau^k_i = t_{i-1}\vee \inf
\biggl\{r: \int_r^{t_{i}} \lambda_{k}(s)ds \le y_i^k
\biggr\}, \; \eta_i^k = {\bf{1}}(\tau^k_i > t_{i-1})  \label{sigma} 
\end{align}
and likelihood contributions
\begin{align} 
f_i^{(k)}( \tau^k_i; \theta) 
 &=  \lambda_{k}(\tau^k_i)^{\eta_i^k} 
\exp\biggl( - \int_{\tau^k_i}^{t_i}
 \lambda_{k}(s) ds \biggr)  \; {\bf 1 }( t_{i-1}\le \tau^k_i < t_i). \label{eq:fi}
\end{align} 

Note that when simulating values $\tau^k_i$, it must be checked whether they lie inside the interval $[t_{i-1}, t_i)$; if not, set $\tau^k_i = t_{i-1}$, under which case we say that truncation has occurred ($\eta_i^k = 0$) and the density $f^{(k)}_i(\tau^k_i;\theta)$ reduces to a point mass probability.

We furthermore designate by $\theta$ the collection of all the unknown parameters of the process $\{X(t)\}$, and we also assume that the realizations of the process are observed for $K$ subjects, $\{X_j(t)\}$ at time points $\{t_{0,j}, t_{1,j}, \dots, t_{n_j, j}\}$, $j = 1, \dots, K$. Then the observed data likelihood for $\theta$ is 
\begin{align}
L(\theta; X) 
& =   \prod_{j=1}^K \gamma( X(t_{0,j})) \prod_{i=1}^{n_j}
P_{X(t_{i-1}) X(t_i)}(t_{i-1,j},t_{i,j}) \nonumber \\
& = \prod_{j=1}^K \gamma( X(t_{0,j})) L_j(\theta; X_j).
\label{Obs:likelihood}
\end{align}

\subsection{Data augmentation and pseudo-marginal McMC algorithm}\label{McMC}
The next step is to use data augmentation with the auxiliary $(\tau^1_{i,j}, \tau^0_{i,j})$ variables in order to construct a pseudo-marginal Markov chain Monte Carlo (McMC) algorithm for parameter estimation.  The conditional likelihood for $\theta$ given $(\tau^1_{i,j}, \tau^0_{i,j})$'s and the observed data is given by 
\begin{align}
\widetilde L(\theta; \tau^1, \tau^0,  X) 
=   \prod_{j=1}^K \gamma( X(t_{0,j})) \prod_{i=1}^{n_j}
\widetilde P_{X(t_{i-1,j}) X(t_{i,j})}
( t_{i-1,j},t_{i,j}) ,\label{Augmented:likelihood}
\end{align}
with the property 
\begin{align*}
\E_{\Pi^1\otimes \Pi^0}\bigl[ \widetilde  L( \theta; \tau^1, \tau^0 ,X) \bigr]
& = L(\theta; X).
\end{align*}

\noindent This precisely shows that the augmented likelihood $\widetilde L( \theta; \tau^1, \tau^0, X)$ is an unbiased estimator of the likelihood $L( \theta; X)$. The Markov chain Monte Carlo (McMC) estimation algorithm which uses this unbiased estimator is referred to as a pseudo-marginal McMC.\\ 

\noindent Note then that the full joint density of $(\theta, \tau^1, \tau^0)$ can be expressed as the product
\begin{align} &
&    \mathcal{L}(\theta, \tau^1,\tau^0; X) =   
    L(\theta; \tau^1, \tau^0, X) \prod_{j=1}^K \prod_{i=1}^{n_i} \prod_{k=0}^1 f_{i,j}^{(k)}(\tau^k_{i,j};\theta).\label{Full:likelihood}
\end{align}

From (\ref{TPM:both}), we can write 
\begin{align*} 
& \widetilde P_{01}(t_{i-1,j},t_{i,j}) =  {\bf 1 }(\tau^0_{i,j} > \tau^1_{i,j}),  \widetilde P_{00}(t_{i-1,j},t_{i,j}) =  {\bf 1 }(\tau^1_{i,j} \ge \tau^0_{i,j} )  \\
& \widetilde P_{10}(t_{i-1,j},t_{i,j}) =  {\bf 1 } (\tau^1_{i,j} > \tau^0_{i,j} ), \widetilde P_{11}(t_{i-1,j},t_i,j)  =  {\bf 1 } (\tau^0_{i,j} \ge  \tau^1_{i,j} ).
\end{align*}

Given $\theta$ and observations, we carry out data augmentation by rejection sampling using the indicator functions giving the TPM. Given the augmented data, the likelihood (\ref{Augmented:likelihood}) is exactly equal to one. Hence, the likelihood function in (\ref{Full:likelihood}) factors into two parts, one involving only $\lambda_{0}(t)$ and the one involving $\lambda_{1}(t)$. To update each rate, we use Metropolis-Hastings algorithm.

Under this developed procedure, the algorithm for the estimation of the parameters $\theta = \big(\lambda_{0}(t), \lambda_{1}(t)\big)$ can be summarized by the following three steps, with Steps 2 and 3 being performed iteratively until convergence. \\ 

\begin{itemize}
    \item[Step 1.] Start with some initial values of $\theta$. Given the initial values and the data, simulate $(\tau^1_i, \tau^0_i)$ repeatedly using (\ref{sigma}) until the pair satisfy one of the inequalities estimating the transition probability for the given interval and states occupied. 
    \item[Step 2.] Update the parameters $\theta$ in a Metropolis-Hastings step using the augmented data, treating $(\tau^1_{i,j}, \tau^0_{i,j})$'s sampled in the previous step as given and fixed, 
    with a posterior proportional to 
    $$\pi(\theta) \prod_{j=1}^K \prod_{i=1}^{n_i} \prod_{k=0}^1 f_{i,j}^{(k)}(\tau^k_{i,j};\theta),$$ where $\pi(\theta)$ is the prior. Note that each rate can be updated separately as the above posterior, assuming independent prior distributions for the two rates can be written as a product of two terms. The log-posterior distributions, for $k=0, 1$, are given as follows. 
    \begin{align*}
    & \log p(\lambda_{k}(\cdot); \tau^1, \tau^0, \mathbf{X}) \propto \log \pi_k(\lambda_{k}()) +  \sum_{j=1}^K \sum_{i=1}^{n_i} \log f_{i,j}^{(k)}(\tau^k_{i,j};\lambda_{k}(\cdot ))  \\
    & \hspace{0.5cm} = \log \pi_k(\lambda_{k}(´\cdot )) + \sum_{j=1}^K \sum_{i=1}^{n_i} \biggl[\eta^k_{i,j} \log \lambda_{k}(\tau^k_{i,j})   - \int_{\tau^k_{i,j}}^{t_{i,j}} \lambda_{k}(s) ds \biggr]. 
    \end{align*}
     \item[Step 3.]
     Using the newly updated values of the $\theta$ from Step 2, update the $(\tau^1_{i,j}, \tau^0_{i,j})$ as described in Step 1. Then update $\theta$ using Step 2.
     \item[Step 4.] Iterate until the convergence.
\end{itemize} 
We point out that sampling of  $(\tau^1_{i,j}, \tau^0_{i,j})$ can be carried out in parallel for each $i$ and $j$ since they are sampled independent of each other. 
 
The choice of initial values of $\theta$ is crucial especially if the observed data are sparse. Also, the choice of prior distributions and the hyperparameters affect the performance of the algorithm. We will discuss these points in connection to constant rates (time homogeneous process) and Weibull rates in the following section.

\section{Example}\label{sec:example_sim}
To illustrate the data augmentation algorithm described in section \ref{McMC}, we consider Weibull-type transition rates which have the functional forms:
\begin{align}
    \lambda_{k}(t) & = \lambda_k \gamma_k t^{\gamma_k - 1}, \; \gamma_k, \lambda_k > 0, t > 0, \label{weibull_trans_0} 
\end{align}
where $\gamma_k$ and $\lambda_k$, $k = 0, 1$ are the shape and rate parameters respectively. 
The transition probabilities of equation (\ref{TPM}) in this case are, for $k \ne l, k, l = 0, 1$:
\begin{align}
    P_{kl}(s,t) 
    & = \exp\biggr( - \lambda_k \; t^{\gamma_k} - \lambda_l \; t^{\gamma_l} \biggl) \gamma_k \lambda_k  
     \int_s^t \exp\biggr(  \lambda_k \; v^{\gamma_k} + \lambda_l \; v^{\gamma_l} \biggl) \; v^{\gamma_k - 1} dv.\nonumber \\
    \label{TPM:Weib} 
\end{align}

Setting $\gamma_0 = \gamma_1 = 1$ gives the constant rates and corresponds to a time homogeneous process.  The above expression simplifies to
\begin{align}
    P_{kl}(s,t) & =\frac{\lambda_k}{\lambda_0 + \lambda_1} \biggl( 1 - \exp\{ - (\lambda_0 + \lambda_1) (t-s)\} \biggr). \label{TPM:EXP}
 \end{align}
We first describe the data augmentation algorithm in the Weibull case assuming independent gamma prior for the four parameters. The density of the Gamma distribution with the shape and rate parameters $(a,b)$ is given as
$$\pi(y;a,b) = \frac{b^a}{\Gamma (a)} y^{a-1} \exp\{- b y \}, \; a, b >0, y > 0.$$
Here $a$ and $b$ can be interpreted as the prior belief on the total number of events and total exposure, respectively. In the absence of prior information, we set both $a_j = b_j = 0.1$.  We will return to the choice of hyperparameters later.

\subsection{Updating the Weibull parameters}\label{sec:update_Weibull}
The integrated Weibull rates are 
\begin{align*}
\int_r^t \lambda_{k}(s)ds =   \lambda_k (t^{\gamma_k} - r^{\gamma_k}), \; k = 0, 1.
\end{align*}
To simplify the notations, we consider one interval $[t_{i-1}, t_i]$ here. The steps described below apply to all intervals $[t_{i-1,j}, t_{i,j}], i = 1, \dots, n_i, j = 1, \dots, N.$ Here $\theta = (\gamma_0, \lambda_0, \gamma_1, \lambda_1).$

\begin{itemize}
    \item[Step 1.]  Start with some initial values of $\theta$ (we will return to this later). Simulate $(\tau^k_i, \eta^k_i), k = 0, 1$ using the rejection sampling described earlier where $\eta^k_{i}   = {\bf 1}(\tau^k_{i} > t_{i-1})$ and 
    \begin{align*}
    \tau^k_{i} & = t_{i-1} \vee \biggr( t_{i}^{\gamma_{k}} - y_i/\lambda_k \biggl)^{1/\gamma_{k}} 
    = \biggr\{t_{i-1}^{\gamma_{k}} \vee ( t_{i}^{\gamma_{k}} - y_i/\lambda_k )\biggl\}^{1/\gamma_{k}}.
    \end{align*}
    \item[Step 2.] Metropolis-Hastings step to update $\theta$. We use generic notations for the shape and rate parameters of the Weibull distribution as $(\gamma, \lambda)$ and the augmented time as $\tau_i \in [t_{i-1},t_i]$. The following is applied separately for $(\gamma_0, \lambda_0, \tau^0_i)$ and  $(\gamma_1, \lambda_1, \tau^1_i)$.

The augmented log-likelihood with $\tau_i \in [t_{i-1},t_i]$ and 
truncation indicators $\eta_i = {\bf 1} (\tau_i > t_{i-1}) \in \{0, 1\}$ is given by 
\begin{align*} &
  \bigl( \log(\lambda) + \log(\gamma) \bigr) \sum_i \eta_i 
 + (\gamma-1)\sum_i \eta_i \log(\tau_i) - {\lambda} \sum_i ( t_i^{\gamma}
 - \tau_i^{\gamma}).
&\end{align*}
We integrate out  $\lambda$ w.r.t. its Gamma$( a,b)$ prior from the above likelihood and obtain a log-marginal density for $\gamma$
\begin{align*} &  \text{const.}+  \log \pi(\gamma) + 
\log( \gamma)\sum_i \eta_i
+\gamma\sum_i \eta_i \log(\tau_i) & \\ &- \biggl( a + \sum_i \eta_i \biggr)
\log\biggl( b +  \sum_i ( t_i^{\gamma}
 - \tau_i^{\gamma}) \biggr) 
&\end{align*}
where the constant does not depend on $\gamma$. 

We assume a Gamma$(\alpha,\beta)$  prior for $\gamma$.  We use a log-normal proposal  with $\log \gamma^* \sim {\mathcal N}(\log\gamma,\sigma^2)$, and update $\gamma$ with Metropolis-Hastings acceptance log-ratio 
\begin{align*} &   
\biggl( \alpha + \sum_i \eta_i \biggr)
\bigl( \log( \gamma^*) - \log( \gamma)  \bigr)
+\biggl( -\beta +\sum_i \eta_i \log(\tau_i) \biggr ) (\gamma^*-\gamma)& \\ &+\biggl( a + \sum_i \eta_i\biggr)
\biggl\{ \log\biggl( b +  \sum_i ( t_i^{\gamma}
 - \tau_i^{\gamma}) \biggr) 
- \log\biggl( b +  \sum_i ( t_i^{\gamma^*}
 - \tau_i^{\gamma^*}) \biggr) \biggr\},
&\end{align*}
 
Then, given $\gamma$, 
$\lambda$ is sampled from its full conditional  distribution
\begin{align*}
\lambda \sim \text{Gamma}\biggl(  a + \sum_i \eta_i, b +   \sum_i ( t_i^{\gamma}
 - \tau_i^{\gamma} ) \biggr).
 \end{align*}
\end{itemize}

In the  constant intensities model, $\gamma = 1$, we  update only $\lambda$ directly from its full conditional distribution. For the sake of completeness, we give detailed pseudo-marginal McMC algorithm for the constant rates in Appendix \ref{appB}. \\

\noindent \textbf{Choice of the initial values.} The pseudo-McMC algorithm can be sensitive to the initial values especially when the observed times are sparse. We suggest to use the crude estimates obtained by treating the sampling points with changes in the observed state as
if they were continuous time  transitions. For the Weibull case, we can assign one as the initial values for the shape parameters and estimate the rate parameter $\lambda_0$ as the ratio of the number of direct transitions from $0$ to $1$ to the sum of lengths of observation intervals starting from state $0$. The initial value of $\lambda_1$ can be obtained correspondingly. We expect that these initial values of the rates provide lower bounds for the maximum likelihood estimates. 

In case of the constant rates model, the estimators of rates, if the observed times were exact transition times, would be the ones that are described above. Since they are not the exact transition times, more (or at least as many as there are observed) transitions within each observed intervals are expected so the numerator would be as large as in the estimator, and the denominator would be smaller or equal to the one in the estimator.  Alternatively, one can use \texttt{msm}  package to obtain estimates of the constant rates for intermittently observed process. Similar suggestions are also discussed in \cite{Titman2011}. \\

\noindent \textbf{Choice of the prior and hyperparameters.} We have chosen a gamma distribution as a prior because of conjugacy so that the posterior distribution of the rate parameter is also a gamma distribution. The improper prior may not be a good choice here because of the sparsely observed data.  \\

\noindent \textbf{Model selection.} 
Approaches to model assessment based on comparison of parametric and nonparametric estimates, model expansion etc. may be problematic in our case. We refer to Section 5.2.3, \cite{Cook2018} for further discussion on this. In Appendix \ref{model_selection},  we briefly describe  the model expansion approach especially for the nested models that are considered here. 

\section{Results}\label{sec:results}

We first display the results for the Weibull rates (\ref{weibull_trans_0})  under two different scenarios; one is based on simulation of a continuous-time Markov process and the other one uses real-life time grids from clinical visits data. In both scenarios, the initial distribution is set to $(0.5, 0.5)$ and  the true shape and rate parameters: $\gamma_{0} = 1.2$, $\gamma_{1} = 0.8$, $\lambda_{0} = 0.006$, and $\lambda_{1} = 0.023$. Furthermore, $5,000$ McMC samples were generated, with a burn-in amount of samples set to $1000$. We present the results under the homogeneous case in Appendix \ref{appE}.

\subsection{Scenario 1: simulated data}\label{sec:simulated}
In this scenario we first simulate continuous-time Markov process with Weibull rates for fixed  number of subjects $N$, the time horizon $T$, and then sample $m$ points discretely from the simulated continuous process to create snapshots in time. The chosen configuration for this particular simulation is $N=100$, $m=50$, and $T=100$. Figure \ref{fig:posterior_plots_weibull_sim} shows the posterior densities of the four parameters and the summary is shown in Table \ref{tab:weibull_results}. It is clear from the plots that the posterior medians of the two shape parameters fall above and below one as expected. The posterior median of $\lambda_1$ is almost at the true value while that of $\lambda_0$ is slightly above the true value. It is to be noted that intermittently observed process may carry different amount of information for each parameter and hence, some parameters may be estimated more precisely than the others. 

\begin{table}[ht]
\caption{Weibull rates. Posterior medians with (95\% credible intervals) of parameters under simulated and real-life time grids. }
\centering
\footnotesize
\begin{tabular}{lccc}
\hline
\textbf{Parameters} & \textbf{True}  & \textbf{Simulated grid} & \textbf{Real-life grid} \\
\hline \\
\shortstack[l]{${\gamma}_0$} & 1.2 & \shortstack[l]{1.14 (0.91, 1.44)} & \shortstack[l]{1.12 (1.01, 1.27)} \\  \\
\shortstack[l]{${\lambda}_0$} & 0.006 & \shortstack[l]{0.008 (0.002, 0.024)} & \shortstack[l]{0.011 (0.004, 0.022)} \\ \\
\shortstack[l]{${\gamma}_1$} & 0.8 & \shortstack[l]{0.85 (0.64, 1.10)} & \shortstack[l]{0.83 (0.69, 0.99)} \\ \\
\shortstack[l]{${\lambda}_1$} & 0.023 & \shortstack[l]{0.022 (0.006, 0.061)} & \shortstack[l]{0.019 (0.006, 0.051)} \\
\hline
\end{tabular}
\label{tab:weibull_results}
\end{table}

\begin{figure}
    \centering
    \includegraphics[width=0.83\linewidth, height=7cm]{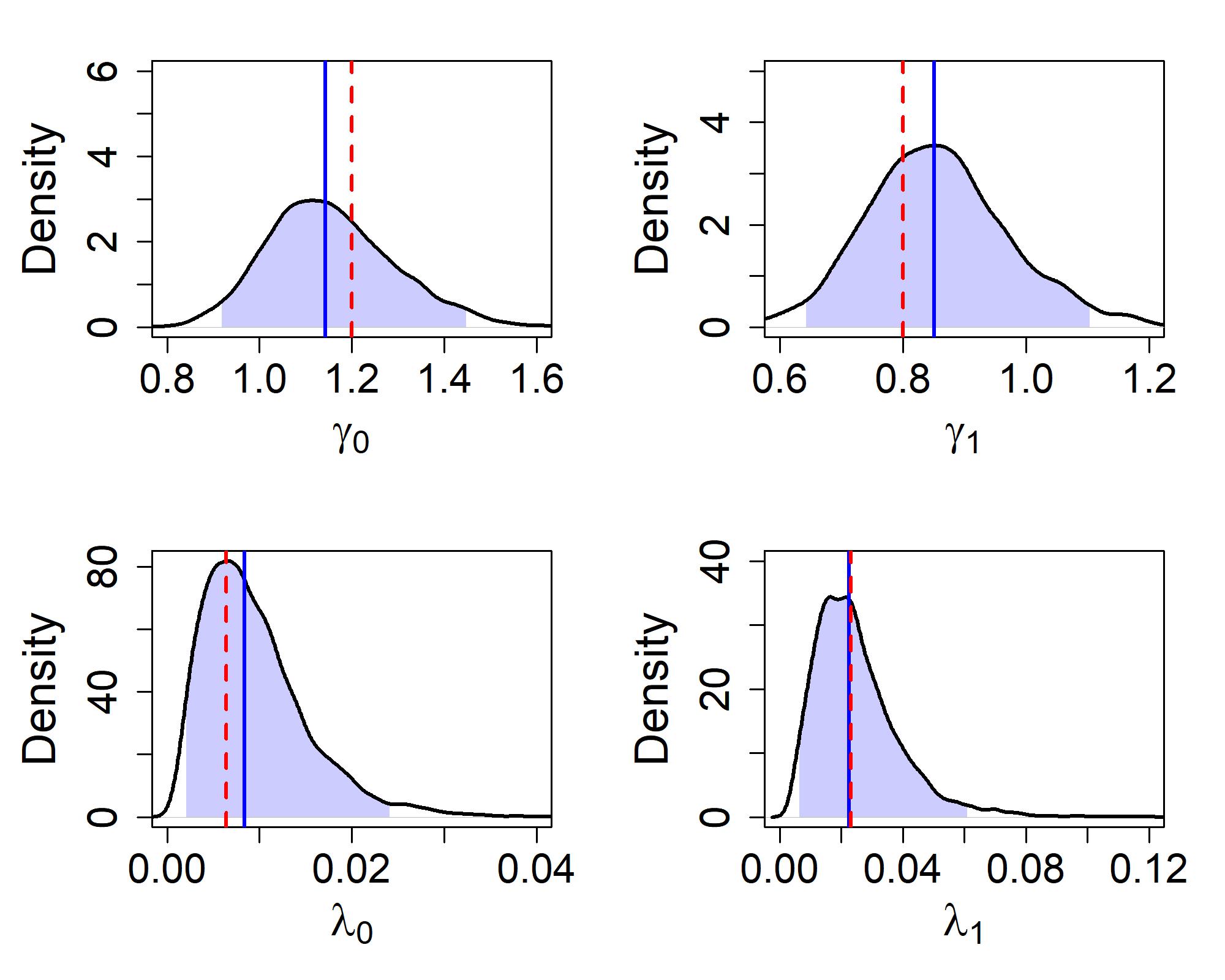}
    \caption{Simulated time grid and Weibull rates. Posterior density plots for the four parameters $\lambda_{0}$, $\lambda_{1}$, $\gamma_{0}$ and $\gamma_{1}$ under the simulated time grid. The dashed red line shows the true value, the solid blue line indicates the posterior median, and the 95\% credible intervals are shaded.}
    \label{fig:posterior_plots_weibull_sim}
\end{figure}

\subsection{Scenario 2: real-life time grid}\label{sec:real}
Here we employ a real-life time grid from $N=93$ subjects according to their clinical visits during $T=48$ months period, and simulate states occupied according to the Weibull rates. The subjects in this scenario have varying number of observations, unlike in our simulated example where $m$ was fixed.  This is what makes real-life clinical data complex, as it introduces irregularity in observation schedules and challenges in aligning individuals on a common assessment timeline. To better understand the real-life grid, we present descriptive statistics in Table \ref{tab:visit_summary}.

\begin{table}[ht]
\centering
\caption{Summary statistics of the number of visits per subject in the real-life time grid.}
\begin{tabular}{llll}
\hline
{Statistic} & Range & Median (Q1, Q3) & Mean (SD) \\
Estimate & (9, 48) & 29 (17, 39) & 29 (12) \\
\hline
\end{tabular}
\label{tab:visit_summary}
\end{table}

Figure \ref{fig:posterior_plots_weibull_real} shows the posterior densities of the four parameters and the summary is shown in Table \ref{tab:weibull_results}. As in section \ref{sec:simulated}, the posterior medians of the two shape parameters fall above and below one but they are farther than those shown in Figure  \ref{fig:posterior_plots_weibull_sim}, as expected. The posterior medians of $\lambda_0$ and $\lambda_1$ are also farther from those for the simulated data. Here also $\lambda_1$ is estimated more precisely than $\lambda_0$. 

\begin{figure}
    \centering
    \includegraphics[width=0.83\linewidth, height=7cm]{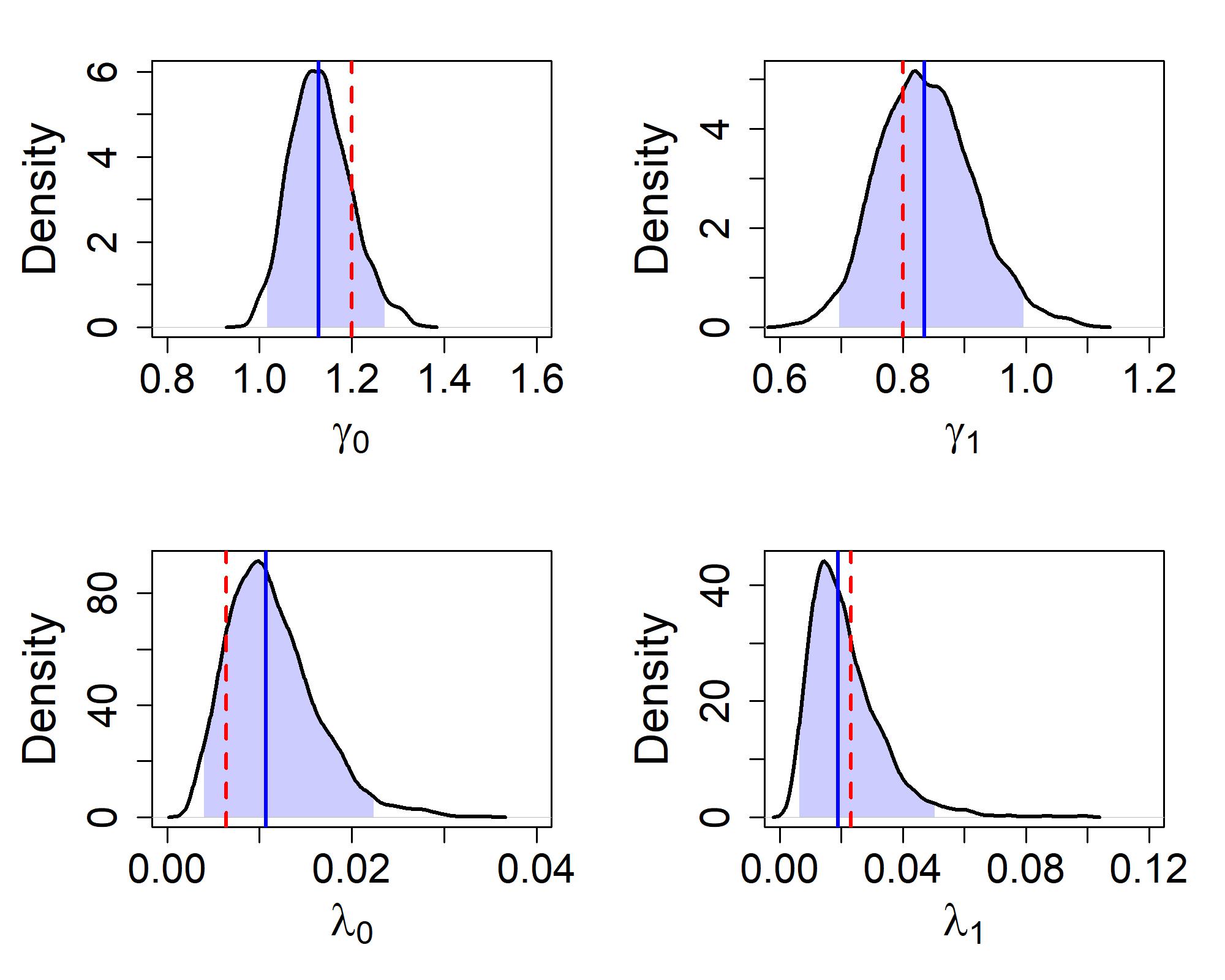}
    \caption{Real-life time grid and Weibull rates. Posterior density plots for the four parameters $\lambda_{0}$, $\lambda_{1}$, $\gamma_{0}$ and $\gamma_{1}$ under the real-life time grid. As before, the dashed red line shows the true value, the solid blue line indicates the estimated value,and the 95\% credible intervals are shaded.}
    \label{fig:posterior_plots_weibull_real}
\end{figure}

\section{Discussion}\label{sec:discussion}

In this paper, we have proposed a very simple and novel pseudo-marginal McMC algorithm which uses a data augmentation procedure with honest times. We implement it to estimate the transition intensities for two-state non-homogeneous Markov models when the trajectories are observed intermittently on irregular and different time grids. Our approach is shown to be very adaptable because it can perform reliable parameter estimation under both homogeneous rates as demonstrated by Tables~\ref{tab:homogeneous_sim_results} and~\ref{tab:homogeneous_sim_results_realgrid} (the results are very similar to the \texttt{msm} package) and non-homogeneous rates (Table~\ref{tab:weibull_results}), irrespective on the  coarseness of  the underlying time grids.
When the integrated PPP intensities have  an analytic form, our algorithm does not introduce
 time-discretization errors as the \texttt{nhm} package \cite{Titman2011} does when solving numerically the linear ODE for the transition matrices.
The \texttt{msm} package \cite{Jackson2011} based on the spectral decomposition or series expansion, should be preferred for time homogeneous Markov models to avoid unnecessary computations needed for non-homogeneous case.\\

The precision with which the model parameters are estimated depends on the sparseness of the observed times, i.e., $t_{i,j} - t_{i-1,j}$. This point is explored in detail for time-homogeneous Markov models in Section 5.2.2, \cite{Cook2018} and the readers interested in further details may refer to it. The derivatives of the transition probabilities and the Fisher information for a general multistate model with intermittent observations are given in the appendix in \cite{Titman2011}. In general, it is difficult to use the Fisher information to comment directly on the precision and the time spacings. The proportions of truncated honest times are related to the information content of the intermittently sampled data. \\ 

Given the simulated time grids, we only have access to a partially sampled trajectory of the process for each subject. As a next step, two research objectives are of interest from a clinical point of view, and our proposed data augmentation algorithm can be extended to meet these objectives. The first aim is to estimate the expected total time spent in state $0$ over the time interval $[0,T]$. The second goal is to be able to predict the sample path for the selected time grid. \\

It is important to acknowledge that, for the purposes of illustration, we have only focused on Weibull-type rates as an example. Nonetheless, we point out that our algorithm also works with other time-dependent rates, such as Gompertz rates. Additionally, our proposed workflow can be extended to multiple correlated non-homogeneous Markov two-state processes. Implications and extensions of the proposed method for multistate models are also of interest.  
Beyond parametric models, in a forthcoming
work we are also combining the data augmentation
via honest times together with a  Bayesian nonparametric model based on
generalized Gamma subordinators (see \cite{Gasbarra2015}).

\bigskip

\noindent \textbf{ACKNOWLEDGMENTS}  \newline Etienne Sebag, was funded by Research Council of Finland (338507). Authors thank Terhi Ollila, Department of Ophthalmology, Helsinki University Hospital, Helsinki, Finland for providing real time grid based on her doctoral study data.

\bibliography{bibliography}


\appendix

\noindent \textbf{Appendix}

\section{Honest times}\label{appC}
 {\bf Definition}
 A random time $\tau$ is an {\it honest time }
 w.r.t. to the filtration $\F=({\mathcal F}_t: t \ge 0)$
 if and only if for every $t > 0$
there exists an ${\mathcal F}_t$-measurable random variable $\tau_t$ such that $\tau = \tau_t$
on the event $\{\tau < t \}$. Then  it is always possible
to choose $\tau_t$ such that
$\tau_t \le \tau$ $\P$-a.s. \cite{jeulin, aksamit}
\\

Intuitively an honest time is the last time
something  happens, so that when it happens we see in our information filtration that  it is happening   without necessarily
knowing  that  it  is happening 
for the last time, and because of that it is not
a stopping time.
\\

For example,
if  $ 0 = \sigma_0\le \sigma_1\le \sigma_2 \le \dots \le \sigma_k \le \dots$ are stopping times 
in the filtration $\F$
and
$T>0$, the ${\mathcal F}_T$-measurable random time
\begin{align*}
 \tau_T = \sup_{k\in \N} 
 \{ \sigma_k {\bf 1}(\sigma_k \le T ) \}
\end{align*}
is an honest time but not necessarily  a stopping time, because
\begin{align*}
\{ \tau_T \le t \} = \bigcap_{k\in \N}
\biggl( \underbrace{  \{ \sigma_k > T\}
}_{\in  {\mathcal F}_T}  \cup
\underbrace{ \{\sigma_k \le t\} }_{\in {\mathcal F}_t }  \biggr)   
\end{align*}
which is not necessarily  in  ${\mathcal F}_t  \text{ when } t <T$.
Moreover  on the event $\{\tau_T \le t\}$,
$\tau_T = \tau_t$,  which
 is ${\mathcal F}_t$-measurable.

\section{Further insight into the two PPPs} 
\label{subsec:appendix} 
Given the two independent PPP with intensity
$\lambda_{0}(t)$ and $\lambda_{1}(t)$ and the initial state $X(s)$, then one can construct the trajectory of $(X(r): s<r\le t)$. The interpretation of the aforementioned construction is as follows. When $X(s)=0$, then the condition $X(t)=1$ holds if, and only if, the random times defined above
satisfy $\tau^0 > \tau^1$, regardless of the number of possible transitions of $X$ between the  states $0$ and $1$ in the interval $(s,t]$. This directly means that, in the interval $(s,t]$, there is at least one $0\rightarrow 1$ transition, 
and after the last $0\rightarrow 1$ transition, there are no more $1\rightarrow 0$ transitions, hence $X(t) = 1$. We therefore interpret $\tau^0$ as the last time point such that there is at least one transition from $0 \rightarrow 1$ in the interval $(s, \tau^0]$ and no such transition between $(\tau^0, t]$. Similarly, $\tau^1$ is the last time point such that there is at least one transition from $1 \rightarrow 0$ in the interval $(s, \tau^1]$ and no such transition between $(\tau^1, t]$. Indeed, for the process to be in state 1 at time $t$, the last transition can not be from 1 to 0 and hence, $\tau^1 < \tau^0$.  This is equivalent to saying that $\tau^1$ is the first time point such that the corresponding PPP has no jump points inside the interval $(\tau^1, t]$. 
We have
\begin{align*} 
& P_{11}(s,t)
  - P_{01}(s,t) =
   P_{00}(s,t) - P_{10}(s,t)  
  =
  \exp\biggl( 
  - \int_s^t \bigl( \lambda_{0}(u)´+
  \lambda_{1}(u) \bigr)du \biggr) 
&\end{align*}
The assumption
\begin{align*}
\int_s^{\infty} \bigl( \lambda_{0}(v)+ \lambda_{1}(v) \bigr)dv 
= \infty
\end{align*}
implies mixing in total variation
\begin{align*}
\lim_{t\to\infty} \sum_{j=1,2}
 |P_{1j}(s,t)- P_{2j}(s,t)|
 =0 \end{align*}  
without  stationarity conditions.  Indeed
for two coupled copies of the
Markov process, $X_t$ and $X'_t$, driven by the 
same pair of independent
Poisson point processes 
$\Pi^1$ and $\Pi^0$ with respective
intensities $\lambda_{1}(t)$
and $\lambda_{0}(t)$,
 started at time
$s$ with
$X_s=1$ 
and $X'_s=0$, the  stopping time
\begin{align*}
\tau = \inf\bigl\{  t: \Pi^1( (s,t] ) +
\Pi^0( (s,t]) >0 \bigr\}
\end{align*}
is a coupling time, such that  $X_t = X'_t, \;   \forall t \ge \tau$, and
 the  total variation distance is given by
$\text{TV}\bigl( X_t, X'_t)= 2 P( \tau > t), \; \forall \;  t > s.$

\section{Constant transition rates and pseudo-marginal McMC algorithm}\label{appB}
Here we describe a pseudo-marginal McMC using both $\tau^1$ and $\tau^0$. 
\begin{enumerate}
    \item Start with some initial values of $(\lambda_0, \lambda_1)$, and then simulate the initial values of $(\tau^1_{i,j}, \tau^0_{i,j})$ as in Step 1, Section \ref{sec:update_Weibull} with $\gamma_0=\gamma_1=1$. Initial value of, for example,  $\lambda_0$, can be obtained as the ratio of the number of intervals starting in $0$ and ending in $1$ to the sum of the interval lengths starting in $0$. We sample the pairs by rejection sampling. \\
    Accept the pair $(\tau^1_{i,j}, \tau^0_{i,j})$ when they satisfy the condition constrained by the states at $t_{i-1,j}$ and $t_{i,j}$. 
    \item For given $(\tau^1_{i,j}, \tau^0_{i,j}), i = 1, \dots, n_j, j = 1, \dots, N$ and $X$, the posterior distribution of  $(\lambda_0, \lambda_1)$ is proportional to the likelihood in~(\ref{Full:likelihood}) and the independent gamma prior. Note that the first term in~(\ref{Full:likelihood})  is one due to the way   $(\tau^1_{i,j}, \tau^0_{i,j}), i = 1, \dots, n_j, j = 1, \dots, N$ are sampled. \\
    Also, the likelihood factorizes into two parts; one involving only $\lambda_0$ and the other one involving only $\lambda_1$. We can write the posterior distribution of $\lambda_k$ by collecting all terms involving $\lambda_k$ as follows.
    \begin{align*}
    & P(\lambda_k \mid \tau^1, \tau^0, \mathbf{X})  
     \propto \lambda_k^{a_k-1} \exp\{- b_k \lambda_k\} \prod_{j=1}^K \prod_{i=1}^{n_i}  f_{i,j}^{(k)}(\tau^k_{i,j};\lambda_k) \\
    & = \lambda_k^{a_k + \sum\limits_j \sum\limits_i \eta^k_{i,j} -1}  \exp\{- \lambda_k (b_k +\sum_j \sum_i (t_{i,j} - \tau^k_{i,j})\},
    \end{align*}
which is a gamma density with          
$(a_k + \sum\limits_{j,i} \eta^k_{i,j})$ and 
$(b_k +\sum\limits_{j,i} (t_{i,j} - \tau^k_{i,j}))$ as the shape and the scale parameters, respectively and $\eta^k_{i,j} = {\bf 1}(\tau^k_{i,j} > t_{i-1,j})$.

\item Simulate values of $(\lambda_0, \lambda_1)$ from the above posterior distributions.
\item Simulate pairs $(\tau^1_{i,j}, \tau^0_{i,j})$ as in step 1.
\item Repeat until the convergence.
\end{enumerate}

\section{Model selection in the class of nested models}\label{model_selection}
We can use the concept of Bayes factor and the model indicator $m = 0, 1.$ In case of exponential and Weibull case, the shape parameter $\gamma$ is equal to one or is different from one. We compare two Bayes models, $\mathcal{M}_m = \{\mathcal{P}_m, \pi^1_m\}$, $m = 0, 1$ (Section 8.2.1, \cite{Robert2025-gg}). The common prior is a mixture of the Dirac-measure at $\gamma=1$ (say $\widetilde\pi^1_0$) and the prior for the continuous $\gamma$ (say $\widetilde\pi^1_1$). Let the prior for the parameters that are not subjected to the Bayes test, rate parameters of the Weibull model, have priors $\pi^1_{kl}(\lambda_{kl}), k, l = 0, 1.$ The Bayes test for testing the point-null hypothesis $H_0: \mathcal{M}_0$ versus $H_1: \mathcal{M}_1$  can be based on the Bayes factor
\begin{align}
B_{01} & = \frac{m_0(\mathbf{X})}{m_1(\mathbf{X})},    
\end{align}
where $m_k(\mathbf{X})$ is the marginal distribution of the data $\mathbf{X}$ given the model $\mathcal{M}_k$ and 
\begin{align*}
m_0(\mathbf{X}) & = \int\limits_{\theta} \P(\mathbf{X} \mid \theta) \pi^1_0(\theta) d\theta \\
& = \int\limits_{\lambda_0, \lambda_1} \P(\mathbf{X} \mid \theta_0) \pi^1_{01}(\lambda_0) \pi^1_{10}(\lambda_1)d\lambda_0 d\lambda_1,
\end{align*}
where $\theta_0 = (\gamma_0=1, \lambda_0, \gamma_1=1, \lambda_1)$. For continuous prior $\pi^0$ of $\theta = (\gamma_0, \lambda_0, \gamma_1, \lambda_1)$ we have the denominator
\begin{align*}
m_1(\mathbf{X}) & = \int\limits_{\theta} \P(\mathbf{X} \mid \theta) \pi^1_1(\theta) d\theta,
\end{align*}
where $\pi^1_1(\theta) = \widetilde\pi^1_1(\gamma_0) \pi^1_{01}(\lambda_0)\widetilde\pi^1_1(\gamma_1) \pi^1_{10}(\lambda_1)$. We have assumed gamma prior distributions for all four parameters. 
We need to compute the marginal likelihoods. For this purpose we can use one or all of the four algorithms described in a tutorial article \cite{GRONAU201780}.

\section{Example: time homogeneous Markov process}\label{appE}
Here, we present the results under the homogeneous case ($\gamma_0 = \gamma_1 = 1)$ for both a simulated time grid (using the same configuration of $N=100$, $m=50$ and $T=100$) as well as the same real-life visiting times of subjects. 

For the simulated process, we consider two sets of rate parameters; one has $\lambda_0$ and $\lambda_1$ fixed close to each other ($\lambda_0 = 0.047$ and $\lambda_1 = 0.051$) and the other where they are set further apart from each other ($\lambda_0 = 0.1$ and $\lambda_1 = 1.0)$. Table \ref{tab:homogeneous_sim_results} shows the results obtained using the proposed data augmentation algorithm (Appendix \ref{appB}) and the \texttt{R} package  \texttt{msm}. The results are similar in both scenarios.

{\renewcommand{\arraystretch}{2}%
\begin{table}[ht]
\centering
\caption{Simulated data for a time homogeneous Markov model. Posterior medians/maximum likelihood estimate with respective 95\% CIs using different estimation methods. Two sets of parameters are considered; one where the two rates are close and the other one where the two rates are far apart.}
\begin{tabular}{lcc}
\hline
Parameter & Data augmentation & msm \\
\hline
\({\lambda}_0\) = 0.047 & 0.048 (0.042, 0.055) & 0.048 (0.042, 0.055) \\
\({\lambda}_1\) = 0.051 & 0.052 (0.045, 0.060) & 0.052 (0.045, 0.060) \\
\hline
\({\lambda}_0\) = 0.100 & 0.103 (0.087, 0.121) & 0.103 (0.087, 0.122) \\
\({\lambda}_1\) = 1.000 & 1.126 (0.967, 1.307)  & 1.134 (0.971, 1.324) \\
\hline
\end{tabular}%
\label{tab:homogeneous_sim_results}
\end{table}
}

We note that estimation of the parameters in the time homegeneous case is possible by directly using the likelihood function in Equation (\ref{Obs:likelihood}). For the real-life time grid, we use ($\lambda_0 = 0.047$ and $\lambda_1 = 0.051$) to simulate the states occupied at each observed time. We then use the exact likelihood, the proposed data augmentation algorithm, 
and the \texttt{R} packages \texttt{msm} and \texttt{nhm}. The results comparing the four estimation approaches are shown in Table \ref{tab:homogeneous_sim_results_realgrid}. The nhm package fails to provide the confidence intervals due to numerical issues relating to the computation of the Hessian. 

{\renewcommand{\arraystretch}{2.2}%
\begin{table}[ht]
\centering
\caption{Real-life time grid. Posterior medians/maximum likelihood estimate with respective 95\% CIs using different estimation methods. (DA = Data Augmentation, Exact = Exact likelihood based estimation)}
\begin{tabular}{lcccc}
\hline
Parameter & DA & msm & nhm & Exact \\
\hline
\({\lambda}_0\) = 0.047 & \shortstack[c]{0.047\\(0.037, 0.061)}
  & \shortstack[c]{0.046\\(0.036, 0.058)} & \shortstack[c]{0.040\\(N/A, N/A)} & \shortstack[c]{0.047\\(0.0378, 0.0610)} \\
\({\lambda}_1\) = 0.051 & \shortstack[c]{0.048\\(0.038, 0.063)} & \shortstack[c]{0.047\\(0.038, 0.060)} & 
\shortstack[c]{0.044\\(N/A, N/A)}  & \shortstack[c]{0.049\\(0.0391, 0.0638)} \\
\hline
\end{tabular}%
\label{tab:homogeneous_sim_results_realgrid}
\end{table}
}

\begin{figure}
    \centering
    \includegraphics[width=0.83\linewidth, height=7cm]{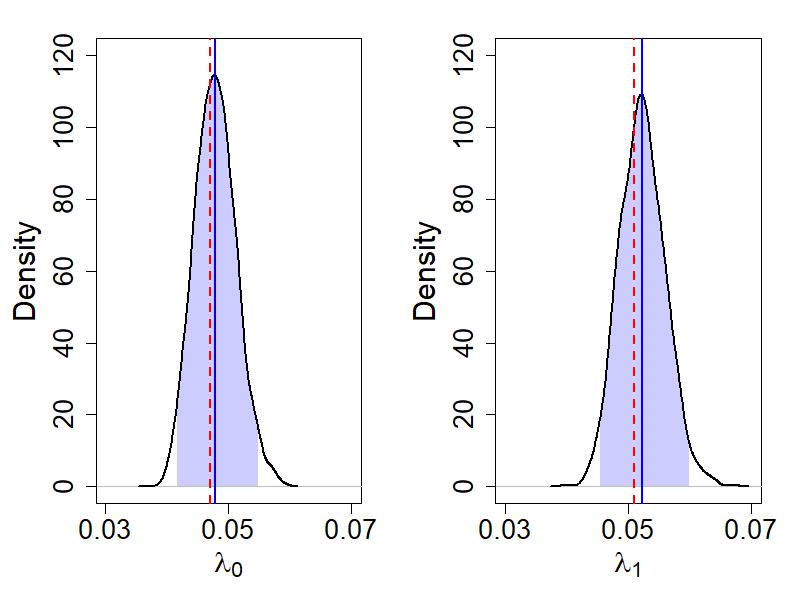}
    \caption{Simulated time grid and the time homogeneous transition intensities. Posterior density plots for $\lambda_{0}$ and $\lambda_1$ under the simulated time grid. The dashed red line shows the true value, the solid blue line indicates the estimated value,and the 95\% credible intervals are shaded.}
    \label{fig:homo_sim_case}
\end{figure}

\begin{figure}
    \centering
    \includegraphics[width=0.83\linewidth, height=7cm]{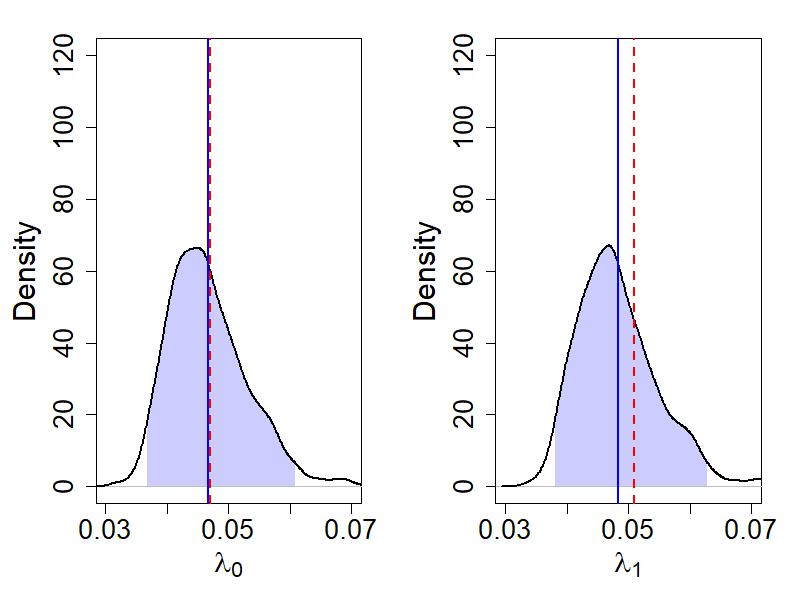}
    \caption{Real-life time grid and the time homogeneous transition intensities. Posterior Density plots  $\lambda_{0}$ and $\lambda_{1}$ under the real-life time grid. The dashed red line shows the true value, the solid blue line indicates the estimated value,and the 95\% credible intervals are shaded.}
    \label{fig:homo_real_case}
\end{figure}

\end{document}